\documentclass[12pt]{article}

\usepackage{amssymb, amsthm}
\usepackage{graphicx}
\usepackage{graphics}
\usepackage{subfig}
\usepackage{epsfig}

\newtheorem{theorem}{Theorem}

\newtheorem{prop}[theorem]{Proposition}

\newtheorem{rem}[theorem]{Remark}

\newtheorem{defn}[theorem]{Definition}

\title{Fractional Periodic Processes: Properties and an Application of Polymer
	Form Factors}

\author{W.~Bock$^1$, J.~L.~da Silva$^2$, L.~Streit$^{2,3}$\\[.3cm]
$^1$ {\small Technische Universit\"at Kaiserslautern,}\\{\small  Fachbereich Mathematik,}\\ {\small Gottlieb-Daimler-Stra{\ss}e 48,}\\ {\small 67663 Kaiserslautern, Germany}\\
{\small E-Mail: bock@mathematik.uni-kl.de}\\[.2cm]
$^2$ {\small CIMA, University of Madeira,} \\{\small Campus da Penteada,} \\{\small 9020-105 Funchal, Portugal.}\\
{\small E-Mail: joses@staff.uma.pt}\\[.2cm]
$^3$ {\small BiBoS, Universit\"at Bielefeld,}\\ {\small Universit\"atsstra{\ss}e 25,}\\{\small  33615 Bielefeld,  Germany}\\
{\small E-Mail: streit@uma.pt}}

\begin{document}

\maketitle

\vspace{10pt}

\begin{abstract}
In this paper we introduce and study three classes of fractional periodic
processes. An application to ring polymers is investigated. We obtain a closed analytic expressions for the form factors, the
Debye functions and their asymptotic decay. The relation between the
end-to-halftime and radius of gyration is computed for these classes
of periodic processes.
\end{abstract}

%
%
%
%
%

\section{Introduction}
Stochastic processes with a periodicity in time have been used e.g.~for
the modelling of stochastic ring structures. In particular in polymer
science ring polymers were based on a periodic random walk.
We shall  consider processes on a half-open interval $[0,L)$, with stationary increments depending only on the geodesic distance along the circle of length $L$, as in eq. (2) below, thus ensuring rotational invariance. It is worth pointing out that the standard Brownian bridge on the interval $[0,L]$ does not have this property. However it is possible to define a Brownian version of these processes as we see below.

In the case of fractional Brownian motion this has been done by Istas
\cite{Istas2005} for the case when the Hurst parameter $H$
is less or equal than the Brownian threshold $H\leq 1/2$. Above this limit
the resulting covariance matrix of the process would no longer be positive semi-definite.  There
exists a recently developed relation between a class of fractional
processes named generalized grey Brownian motion and the class of
fractional Brownian motions which is given by multiplying the latter
with a certain time-independent random variable. We can thus define a periodic generalized grey Brownian motion using
this relation. 

The paper is organized as follows. In Section
2 we define three classes of real-valued periodic processes.
We note that the extension to $\mathbb{R}^{d}$-versions valued of these classes
is straightforward. Section 3 is dedicated to the form factors of
these processes. As mentioned before concrete applications are in long
range coupled polymer models. The Debye function and the form factors
are well known quantities from scattering theory and polymer physics.
They give a deeper insight into the scaling behavior of observables
linked with the underlying processes. We derive, for all three classes,
analytic expressions of the form factors. In the appendix we recall
some special functions used in this paper.

\section{Classes of Periodic Processes}

\label{sec:classes-processes}In this section we introduce the classes
of processes used in this paper. They are \emph{periodic processes} with ``time'' parameter $t$ varying
on the circle $\mathbb{S}_{L}$ of length $L>0$. We parametrize the
points on the circle $\mathbb{S}_{L}$ by their angles $\theta\in[0,2\pi)$. Fixing the length $L$ of the
circle $\mathbb{S}_{L}$, then we may parametrize the points on $\mathbb{S}_{L}$
as $t\in[0,L)$.

We assume given a complete probability space $(\Omega,\mathcal{F},P)$ for any of these processes.

\subsection{Periodic fractional Brownian motion}

A \emph{periodic fractional Brownian motion} (pfBm for short) $B_{p}^{H}$
with Hurst parameter $0<H\le 1/2$ is a centered Gaussian process indexed
by $\mathbb{S}_{L}$ with covariance function $R^{H}$ given, for
any $0\leq s,t<L$, by 
\begin{equation}
R^{H}(t,s;L):=\frac{1}{2}\big(d_{H}(t;L)+d_{H}(s;L)-d_{H}(t-s;L)\big),\label{eq:cov_pfBm}
\end{equation}
where 
\begin{equation}
d_{H}(\tau;L):=\min\left\{ |\tau|^{2H},(L-|\tau|)^{2H}\right\} .\label{eq:distance}
\end{equation}
The existence of these Gaussian processes is based on the positive semi-definiteness of the above covariance matrix, as argued by Istas \cite{Istas2005} and \cite[Chap.~VIII]{Levy48}, note however \cite{BBS19} concerning the restriction to $0<H\leq 0.5$.

\begin{rem}
For any $0\le s,t<L$, the variance of the increment $B_{p}^{H}(t)-B_{p}^{H}(s)$
follows from (\ref{eq:cov_pfBm}) and (\ref{eq:distance}) and we
have 
\begin{equation}
\mathbb{E}\Big(\big(B_{p}^{H}(t)-B_{p}^{H}(s)\big)^{2}\Big)=d_{H}(t-s;L).\label{eq:variance-inc-pfBm}
\end{equation}
In addition, the characteristic function of $B_{p}^{H}(t)-B_{p}^{H}(s)$
is 
\begin{equation}
\mathbb{E}\left(e^{ik(B_{p}^{H}(t)-B_{p}^{H}(s))}\right)=\exp\left(-\frac{k^{2}}{2}d_{H}(t-s;L)\right).\label{eq:cf-inc-pfBm}
\end{equation}
\end{rem}

\begin{prop}
\label{prop:pfBm-sssi}
\begin{enumerate}
\item The pfBm process is $H$-self-similar with stationary increments.
\item The pfBm process has a continuous modification. For any $\gamma\in(0,H)$
this modification is $\gamma$-H{\"o}lder continuous on each finite
interval.
\end{enumerate}
\end{prop}

\begin{proof}
1. The $H$-self-similarity is expressed as the following equality
in finite-dimensional distribution, for any $a>0$ and any $t\in(0,L)$
we have 
\[
B_{p}^{H}(at)=a^{H}B_{p}^{H}(t).
\]
This equality can be translated in terms of the covariance function
$R^{H}$, more precisely if $0\le s,t<L$ it holds 
\[
R^{H}(at,as;aL)=a^{2H}R^{H}(t,s;L).
\]
Hence, we have 
\[
R^{H}(at,as;aL)=\frac{1}{2}\big(d_{H}(at;aL)+d_{H}(as;aL)-d_{H}(at-as;aL)\big)
\]
and it is easy to see from (\ref{eq:distance}) that 
\begin{eqnarray*}
d_{H}(at;aL) & =a^{2H}d_{H}(t;L),\\
d_{H}(as;aL) & =a^{2H}d_{H}(s;L),\\
d_{H}(at-as;aL), & =a^{2H}d_{H}(t-s;L).
\end{eqnarray*}
Then the $H$-self-similarity of $B_{p}^{H}$ follows easily. To prove
the stationarity of the Gaussian process $B_{p}^{H}$ it is sufficient
to show that 
\begin{equation}
\mathbb{E}\left((B_{p}^{H}(t)-B_{p}^{H}(s))^{2}\right)=\mathbb{E}\left((B_{p}^{H}(t-s))^{2}\right).\label{eq:pfBm_stationary}
\end{equation}
Equality (\ref{eq:pfBm_stationary}) is a consequence of (\ref{eq:variance-inc-pfBm})
and the definition of $d_{H}$, more precisely we have 
\[
\mathbb{E}\left((B_{p}^{H}(t)-B_{p}^{H}(s))^{2}\right)=d_{H}(t-s;L)=\mathbb{E}\left((B_{p}^{H}(t-s))^{2}\right).
\]
2. Since $B_{p}^{H}(t)-B_{p}^{H}(s)$ is a centered Gaussian random
variable with variance $d_{H}(|t-s|;L)$ we have 
\[
\mathbb{E}(|B_{p}^{H}(t)-B_{p}^{H}(s)|^{p})=(d_{H}(|t-s|;L))^{p/2}\le|t-s|^{pH}.
\]
Thus, if we take $p>\frac{1}{H}$ we obtain the existence of a continuous
modification via the Kolmogorov-Chentsov continuity theorem. For the H{\"o}lder
exponent one obtain $\gamma\in(0,H-\frac{1}{p})$.
\end{proof}
From here on we work with the continuous modification of $B_{p}^{H}$
preserving the same notation.

\subsection{Periodic Grey Brownian Motion}

\label{subsec:pgBm}Grey Brownian motion was introduced by W. Schneider
in \cite{Schneider90,MR1190506} in order to study the fractional
Feynman-Kac formula. Here we are interested in the periodic version
of this process which is represented in terms of the pfBm process
$B_{p}^{H}$ as 
\[
X_{p}^{H}(t):=\sqrt{Y_{2H}}B_{p}^{H}(t),\qquad0\le t<L,\;0<H\le\frac{1}{2},
\]
where $Y_{2H}$ is the positive random variable, independent of $B_{p}^{H}$,
with density given via the $M$-Wright function $M_{2H}$, see Appendix\ \ref{sec:appendix}-(\ref{eq:M-Wright-function}).
We call this process \emph{periodic grey Brownian motion} (pgBm for
short), see Remark\ \ref{rem:pgBm} for more details.

It is easy to compute the characteristic function of the increment
$X_{p}^{H}(t)-X_{p}^{H}(s)$, $0\le s,t<L$, namely, 
\[
\mathbb{E}\left(e^{ik(X_{p}^{H}(t)-X_{p}^{H}(s))}\right)=\int_{0}^{\infty}M_{2H}(\tau)\mathbb{E}\left(e^{ik\sqrt{\tau}(B_{p}^{H}(t)-B_{p}^{H}(s))}\right)\,d\tau
\]
and using equality (\ref{eq:cf-inc-pfBm}) we obtain 
\[
\mathbb{E}\left(e^{ik(X_{p}^{H}(t)-X_{p}^{H}(s))}\right)=\int_{0}^{\infty}M_{2H}(\tau)e^{-\frac{k^{2}}{2}\tau d_{H}(t-s;L)}\,d\tau.
\]
Finally, using the Laplace transform of the density $M_{2H}$, see
(\ref{eq:LaplaceT_MWf}), we arrive at 
\begin{equation}
\mathbb{E}\left(e^{ik(X_{p}^{H}(t)-X_{p}^{H}(s))}\right)=E_{2H}\left(-\frac{k^{2}}{2}d_{H}(t-s;L)\right).\label{eq:cf-inc-pgBm}
\end{equation}
Here $E_{2H}$ is the Mittag-Leffler function defined in (\ref{eq:MLf}).
It follows from (\ref{eq:cf-inc-pgBm}) that 
\begin{equation}
\mathbb{E}\big((X_{p}^{H}(t)-X_{p}^{H}(s))^{2}\big)=\frac{d_{H}(t-s;L)}{\Gamma(2H+1)}.\label{eq:2nd-moment-inc-pgBm}
\end{equation}

\begin{rem}
\label{rem:pgBm}For $\beta=2H$ the process $X_{p}^{\frac{\beta}{2}}$
has characteristic function, for any $\lambda\in\mathbb{R}^{n}$ and
$0\le t_{1}<t_{2}<\ldots<t_{n}<L$ 
\[
\mathbb{E}\left(\exp\left(i\sum_{k=1}^{n}\lambda_{k}X_{p}^{\frac{\beta}{2}}(t_{k})\right)\right)=E_{\beta}\left(-\frac{1}{2}(\lambda,\Sigma\lambda)\right),
\]
where $\Sigma=(a_{kj})_{1\le k,j\le n}$, $a_{kj}=\mathbb{E}\big(X_{p}^{\frac{\beta}{2}}(t_{k})X_{p}^{\frac{\beta}{2}}(t_{j})\big)$
is the covariance matrix and $(\cdot,\cdot)$ denotes the inner product
in $\mathbb{R}^{n}$. When the parameter $t$ is interpreted as time,
that is $t\in\mathbb{R}_{+}$ or a subset $I\subset\mathbb{R}_{+}$,
the process $\sqrt{Y_{\beta}}B^{\frac{\beta}{2}}(t)$, $t\in\mathbb{R}_{+}$
is the one of W.\ Schneider \cite{MR1190506,Schneider90}.
A systematic study of this class of processes and its generalization
was realized by F.\ Mainardi and his collaborators, see \cite{Mura_mainardi_09}
and references therein.
\end{rem}

\begin{rem}
Actually we could start by giving the characteristic functional as
in (\ref{eq:cf-inc-pgBm}) and show the conditions of Minlos-Sazonov's
theorem. This approach leads to the Mittag-Leffler analysis (see \cite{GJRS14})
where the law of the process is a probability measure on a space of
generalized functions.
\end{rem}

\begin{prop}
\label{prop:pgBm-sssi}The process $X_{p}^{H}$ is $H$-self-similar
with stationary increments and continuous paths.
\end{prop}

\begin{proof}
The proof follows as in Proposition\ \ref{prop:pfBm-sssi}.
\end{proof}

\subsection{Periodic Generalized Grey Brownian Motion}

\label{subsec:pggBm}We finally introduce the most general class of
periodic processes used in this paper. More precisely, let $X_{p}^{\beta,H}$
denote the process defined by 
\[
X_{p}^{\beta,H}(t):=\sqrt{Y_{\beta}}B_{p}^{H}(t),\quad t\ge0,\;0<\beta\le1,\;0<H\le\frac{1}{2},
\]
where $B_{p}^{H}$ is the pfBm and $Y_{\beta}$ is the positive random
variable, independent of $B_{p}^{H}$, with density $M_{\beta}$.
We call this process \emph{periodic generalized grey Brownian motion}.
The characteristic function of the increment $X_{p}^{\beta,H}(t)-X_{p}^{\beta,H}(s)$,
$0\le s,t<L$ may be computed as 
\[
\mathbb{E}\left(e^{ik(X_{p}^{\beta,H}(t)-X_{p}^{\beta,H}(s))}\right)=\int_{0}^{\infty}M_{\beta}(\tau)\mathbb{E}\left(e^{ik\sqrt{\tau}(B_{p}^{H}(t)-B_{p}^{H}(s))}\right)\,d\tau
\]
and using equality (\ref{eq:cf-inc-pfBm}) we obtain 
\[
\mathbb{E}\left(e^{ik(X_{p}^{\beta,H}(t)-X_{p}^{\beta,H}(s))}\right)=\int_{0}^{\infty}M_{\beta}(\tau)e^{-\frac{k^{2}}{2}\tau d_{H}(t-s;L)}\,d\tau.
\]
Finally, using the Laplace transform of the density $M_{\beta}$,
see (\ref{eq:LaplaceT_MWf}), we obtain 
\begin{equation}
\mathbb{E}\left(e^{ik(X_{p}^{\beta,H}(t)-X_{p}^{\beta,H}(s))}\right)=E_{\beta}\left(-\frac{k^{2}}{2}d_{H}(t-s;L)\right).\label{eq:cf-inc-pggBm}
\end{equation}
It follows from (\ref{eq:cf-inc-pggBm}) that 
\begin{equation}
\mathbb{E}\big((X_{p}^{\beta,H}(t)-X_{p}^{\beta,H}(s))^{2}\big)=\frac{d_{H}(t-s;L)}{\Gamma(\beta+1)}.\label{eq:2nd-moment-inc-pggBm}
\end{equation}

\begin{prop}
The process $X_{p}^{\beta,H}$ is $H$-self-similar with stationary
increments and continuous paths.
\end{prop}

\begin{proof}
The proof is similar to that of Proposition\ \ref{prop:pfBm-sssi}.
\end{proof}

\section{Form Factors for Periodic Processes}

\label{sec:Form_Factors} In this section we explore the form factors and the corresponding Debye functions for the classes of periodic processes introduced above.
Explicit analytic expressions are computed for all three classes
of periodic processes. The relation between the radius of gyration
and end-to-halftime length is also shown.

\subsection{Form Factors for Periodic Fractional Brownian Motion}

To begin with we note that, given a $d$-dimensional stochastic process
$X$, the form factor associated to $X$ is the function defined by
\begin{equation}
S^{X}(k):=\frac{1}{n^{2}}\int_{0}^{n}dt\int_{0}^{n}ds\,\emph{E}\big(e^{i(k,X(t)-X(s))}\big),\quad k\in\mathbb{R}^{d},\;n\in\mathbb{N}\label{eq:form_factor_general}, see  e.g.~\cite{Hammouda}
\end{equation}
which, in case $X$ is $\nu$-self-similar, simplifies to 
\begin{equation}
S^{X}(k)=\int_{0}^{1}dt\int_{0}^{1}ds\,\emph{E}\big(e^{in^{\nu}(k,X(t)-X(s))}\big).\label{eq:form_factor_general1}
\end{equation}
This function encodes in particular, to lowest order in $k$, the
\emph{root-mean-square radius of gyration} (or simply \emph{radius
of gyration}) of $X$, defined by 
\[
\left(R_{g}^{X}\right)^{2}:=\frac{1}{2}\frac{1}{n^{2}}\int_{0}^{n}dt\int_{0}^{n}ds\,\emph{E}\left(\big|X(t)-X(s)\big|^{2}\right)
\]
which plays an important role in the study of random path conformations.

Hence, for the class of pfBm, denoting the form factor by $S^{\mathrm{pfBm}}$,
for any $k\in\mathbb{R}$, we have 
\begin{eqnarray*}
S^{\mathrm{pfBm}}(k) & :=\frac{2}{L^{2}}\int_{0}^{L}\int_{0}^{t}\mathbb{E}\left(e^{ik(B_{p}^{H}(t)-B_{p}^{H}(s))}\right)\,ds\,dt\\
 & =\frac{2}{L^{2}}\int_{0}^{L}\int_{0}^{t}e^{-\frac{1}{2}k^{2}d_{H}(t-s;L)}\,ds\,dt.
\end{eqnarray*}
Making the
change of variable $\tau=t-s$ we obtain 
\begin{eqnarray*}
S^{\mathrm{pfBm}}(k) & =\frac{2}{L^{2}}\int_{0}^{L}(L-\tau)e^{-\frac{1}{2}k^{2}d_{H}(\tau;L)}\,d\tau\\
 & =\frac{2}{L^{2}}\left(\int_{0}^{L/2}(L-\tau)e^{-\frac{1}{2}k^{2}\tau^{2H}}\,d\tau+\int_{L/2}^{L}(L-\tau)e^{-\frac{1}{2}k^{2}(L-\tau)^{2H}}\,d\tau\right)\\
 & =\frac{2}{L}\int_{0}^{L/2}e^{-\frac{1}{2}k^{2}\tau^{2H}}\,d\tau\\
 & =\frac{(\frac{1}{2}k^{2})^{-\frac{1}{2H}}}{LH}\left[\Gamma\left(\frac{1}{2H}\right)-\Gamma\left(\frac{1}{2H},\frac{1}{2}k^{2}\left(\frac{L}{2}\right)^{2H}\right)\right],
\end{eqnarray*}
where $\Gamma(a,z)=\int_{z}^{\infty}t^{a-1}e^{-t}\,dt$ is the upper
incomplete gamma function, see \cite{GR81}.

Denoting $y^{2}:=\frac{k^{2}}{2}\left(\frac{L}{2}\right)^{2H}$ the
Debye function for pfBm has the explicit expression 
\[
S^{\mathrm{pfBm}}(k)=f^{\mathrm{pfBm}}(y;H):=\frac{1}{2Hy^{\frac{1}{H}}}\left[\Gamma\left(\frac{1}{2H}\right)-\Gamma\left(\frac{1}{2H},y^{2}\right)\right].
\]
In Figure\ \ref{fig:forms_factors_pfBm} we plot the Debye function
$f^{\mathrm{pfBm}}$ corresponding to the pfBm $B_{p}^{H}$ for different
Hurst parameters $H$. The asymptotic of the Debye function $f^{\mathrm{pfBm}}$
is given by 
\[
f^{\mathrm{pfBm}}(y;H)\sim\Gamma\left(\frac{1}{2H}\right)\frac{y^{-1/H}}{2H},\quad y\to\infty
\]
and for $H=\frac{1}{2}$, i.e.~in the case of periodic Brownian motion, $f^{\mathrm{pfBm}}$ decays as $y^{-2}$.
In general $f^{\mathrm{pfBm}}$ decays as $y^{-\frac{1}{H}}$, see
Figure \ref{fig:forms_factors_pfBm}-\subref{fig:LogLog-scale}
where the lines are getting steeper for smaller $H$.

\begin{figure}
\begin{centering}
\subfloat[\label{fig:Linear-scale}Linear scale.]{\begin{centering}
\includegraphics[scale=0.43]{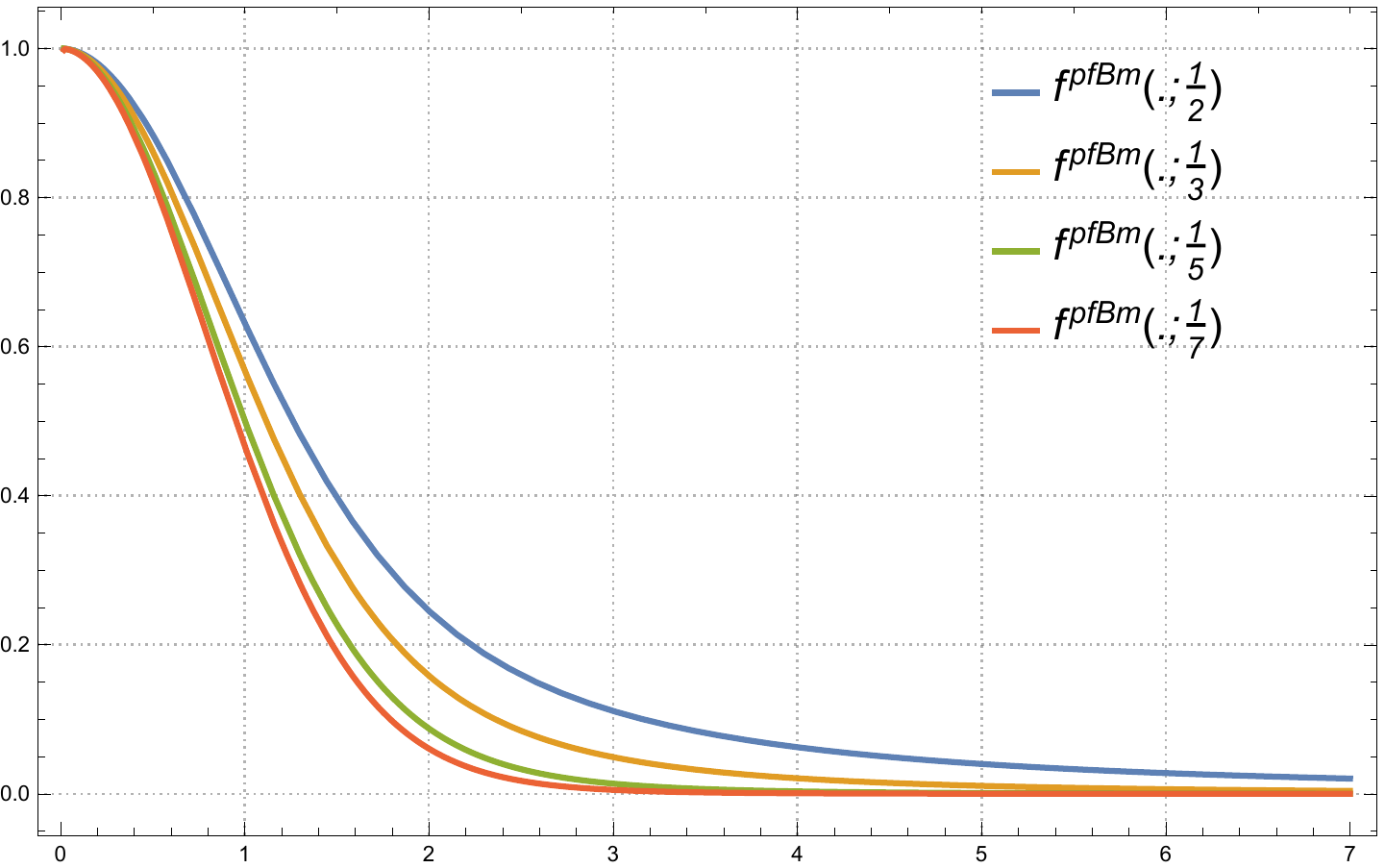}
\par\end{centering}
}\hfill{}\subfloat[\label{fig:LogLog-scale}LogLog scale.]{\begin{centering}
\includegraphics[scale=0.43]{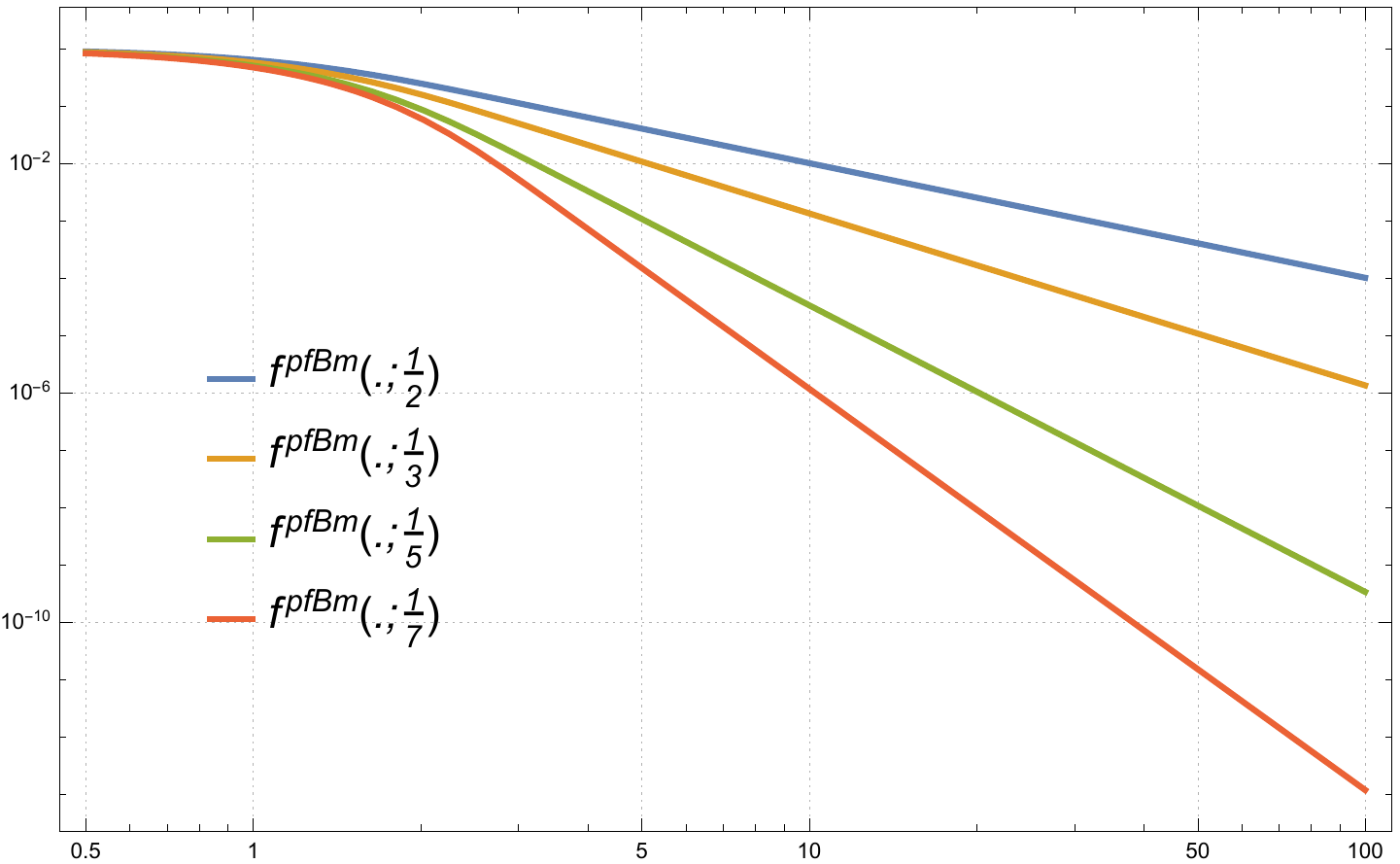}
\par\end{centering}
}
\par\end{centering}
\caption{\label{fig:forms_factors_pfBm}Debye function for the pfBm process
$B_{p}^{H}$ for $H=\frac{1}{2},\frac{1}{3},\frac{1}{5},\frac{1}{7}$.}
\end{figure}

For the the illustration of the dependence on large $y$, Figure \ref{fig:Kratky_fbm} shows the Kratky plot for pfBm. Here the asymptotical behavior is very well visible. 

\begin{figure}
\begin{centering}
\includegraphics[width=.5\textwidth]{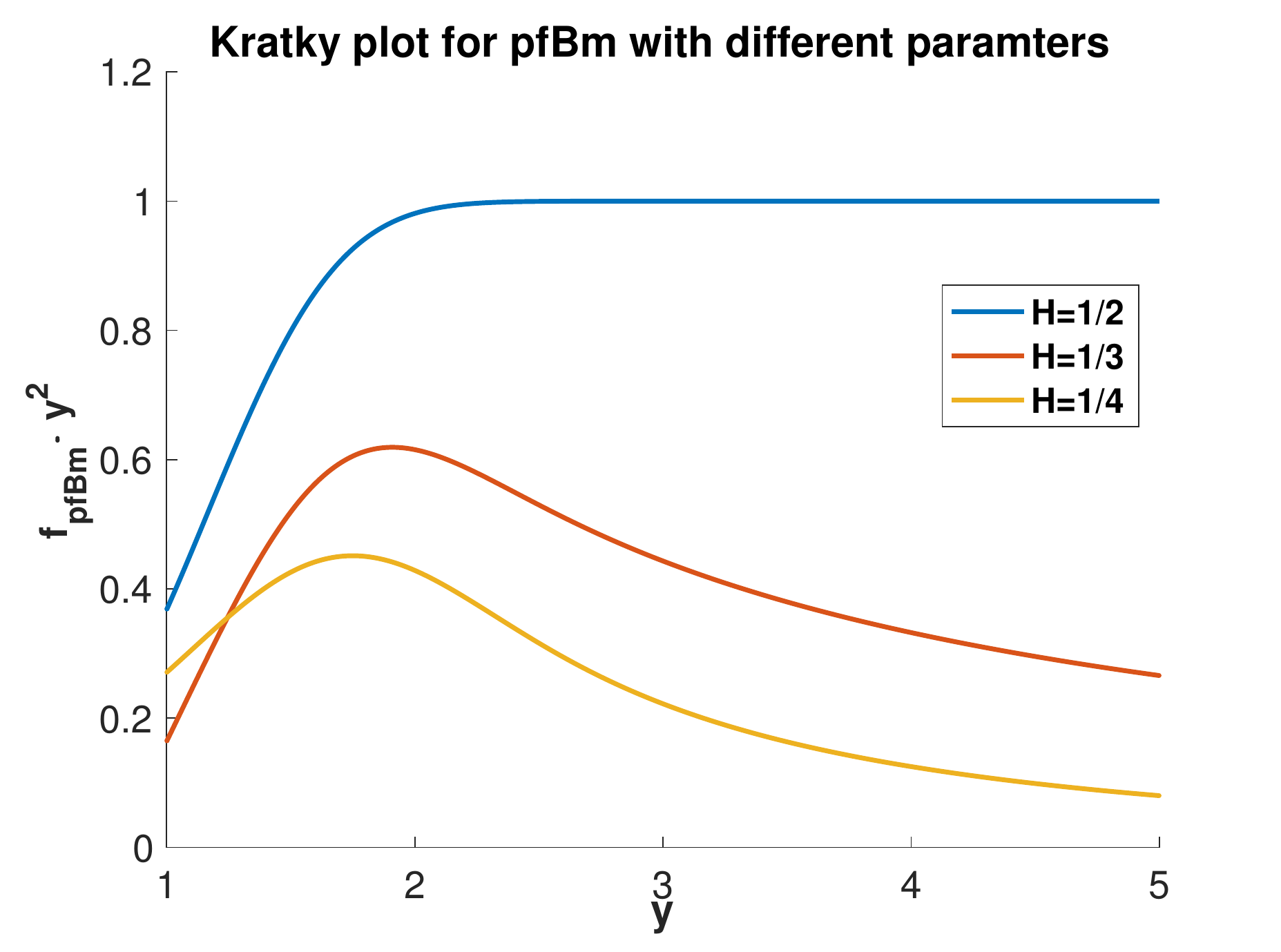}
\caption{\label{fig:Kratky_fbm}Kratky plot for the pfBm process
$X_{p}^{H}$ for $H=\frac{1}{2},\frac{1}{3},\frac{1}{4},\frac{1}{5}$.}
\end{centering}
\end{figure}

The radius of gyration for pfBm is obtained by expanding the form
factor $S^{\mathrm{pfBm}}$ to lowest order. We obtain 
\[
\left(R_{g}^{\mathrm{pfBm}}(L)\right)^{2}\ =\frac{L^{2H}}{(2H+1)2^{2H+1}}
\]
to be compared with the linear case with time parameter $t\in[0,l]$
\[
\left(R_{g}^{\mathrm{fBm}}(l)\right)^{2}=\frac{l^{2H}}{(2H+1)(2H+2)}.
\]
Note that 
\[
\left(R_{g}^{\mathrm{fBm}}\left(\frac{L}{2}\right)\right)^{2}=\frac{\left(R_{g}^{\mathrm{pfBm}}(L)\right)^{2}}{H+1}.
\]
Computing the end-to-halftime length with time parameter $t\in[0,\frac{L}{2}]$
gives 
\[
\left(R_{e}^{\mathrm{pfBm}}\left(\frac{L}{2}\right)\right)^{2}=\mathbb{E}\left(\left(B_{p}^{H}\left(\frac{L}{2}\right)\right)^{2}\right)=d_{H}\left(\frac{L}{2};L\right)=\left(\frac{L}{2}\right)^{2H}
\]
which implies the following relation 
\[
\frac{\left(R_{e}^{\mathrm{pfBm}}\left(\frac{L}{2}\right)\right)^{2}}{2(2H+1)}=\left(R_{g}^{\mathrm{pfBm}}(L)\right)^{2}.
\]

\subsection{Form Factors for Periodic Grey Brownian Motion}

The form factor of the pgBm process $X_{p}^{H}$ introduced in Subsection\ \ref{subsec:pgBm}
is given for any $k\in\mathbb{R}$ by 
\begin{eqnarray*}
S^{\mathrm{pgBm}}(k): & =\frac{2}{L^{2}}\int_{0}^{L}\int_{0}^{t}\mathbb{E}\left(e^{ik\sqrt{Y_{2H}}(B_{p}^{H}(t)-B_{p}^{H}(s))}\right)\,dsdt\\
 & =\frac{2}{L^{2}}\int_{0}^{L}\int_{0}^{t}\int_{0}^{\infty}M_{2H}(\tau)\mathbb{E}\left(e^{-\frac{1}{2}k^{2}\sqrt{\tau}(B_{p}^{H}(t)-B_{p}^{H}(s))}\right)\,d\tau\,ds\,dt.\\
 & =\frac{2}{L^{2}}\int_{0}^{L}\int_{0}^{t}\int_{0}^{\infty}M_{2H}(\tau)e^{-\frac{1}{2}k^{2}\tau d_{H}(t-s;L)}\,d\tau\,ds\,dt.
\end{eqnarray*}
Applying the Fubini theorem and making the change of variables $r=t-s$,
yields 
\begin{eqnarray*}
S^{\mathrm{pgBm}}(k) & =\frac{2}{L^{2}}\int_{0}^{\infty}M_{2H}(\tau)\int_{0}^{L}\int_{0}^{t}e^{-\frac{1}{2}k^{2}\tau d_{H}(t-s;L)}\,d\tau\,ds\,dt\\
 & =\frac{2}{L^{2}}\int_{0}^{\infty}M_{2H}(\tau)\int_{0}^{L}(L-r)e^{-\frac{1}{2}k^{2}\tau d_{H}(r;L)}\,dr\,d\tau.
\end{eqnarray*}
Once more Fubini's theorem gives 
\begin{eqnarray*}
S^{\mathrm{pgBm}}(k) & =\frac{2}{L^{2}}\int_{0}^{L}(L-r)\int_{0}^{\infty}M_{2H}(\tau)e^{-\frac{1}{2}k^{2}\tau d_{H}(r)}\,d\tau\,dr\\
 & =\frac{2}{L^{2}}\int_{0}^{L}(L-r)E_{2H}\left(-\frac{k^{2}}{2}d_{H}(r;L)\right)\,dr\\
 & =\int_{0}^{1}E_{2H}\left(-\frac{k^{2}}{2}\left(\frac{L}{2}\right)^{2H}r^{2H}\right)\,dr\\
 & =E_{2H,2}\left(-\frac{k^{2}}{2}\left(\frac{L}{2}\right)^{2H}\right).
\end{eqnarray*}
In the last equality we have used formula (\ref{eq:Euler_transform-MLf1})
with $\rho=\sigma=\alpha=1$, $\gamma=2H$ and $E_{\beta,\alpha}$
is the generalized Mittag-Leffler, see (\ref{eq:gMLf}). Defining
$y^{2}:=\frac{k^{2}}{2}\left(\frac{L}{2}\right)^{2H}$ it follows
that the Debye function associated to the pgBm process $X_{p}^{H}$
is explicitly given by 
\[
f^{\mathrm{pgBm}}(y;H)=E_{2H,2}\left(-y^{2}\right).
\]
The asymptotic of $f^{\mathrm{pgBm}}(\cdot;H)$ follows from Proposition\ \ref{prop:asymp-gMLf}-\ref{asym-gMLf-two}
such that 
\begin{equation}
f^{\mathrm{pgBm}}(y;H)\sim\frac{1}{\Gamma(2H+2)}\frac{1}{y^{2}},\quad y\to\infty.\label{eq:Debye-pgBm-asym}
\end{equation}
In Figure\ \ref{fig:forms_factors_Streit} we show the Debye function
of the pgBm process $X_{p}^{H}$ for different values of the parameter
$H$, in linear scale Figure \ref{fig:forms_factors_Streit}-\subref{fig:Linear-scale-1}
and LogLog scale Figure \ref{fig:forms_factors_Streit}-\subref{fig:LogLog-scale-1}.
The asymptotic of the Debye function $f^{\mathrm{pgBm}}(\cdot;H)$
given in (\ref{eq:Debye-pgBm-asym}) is reflected in Figure \ref{fig:forms_factors_Streit}-\subref{fig:LogLog-scale-1}
where the slope of the lines for $y$ big being the same. In other
words, the lines are parallel.

\begin{figure}
\begin{centering}
\subfloat[\label{fig:Linear-scale-1}Linear scale.]{\begin{centering}
\includegraphics[scale=0.43]{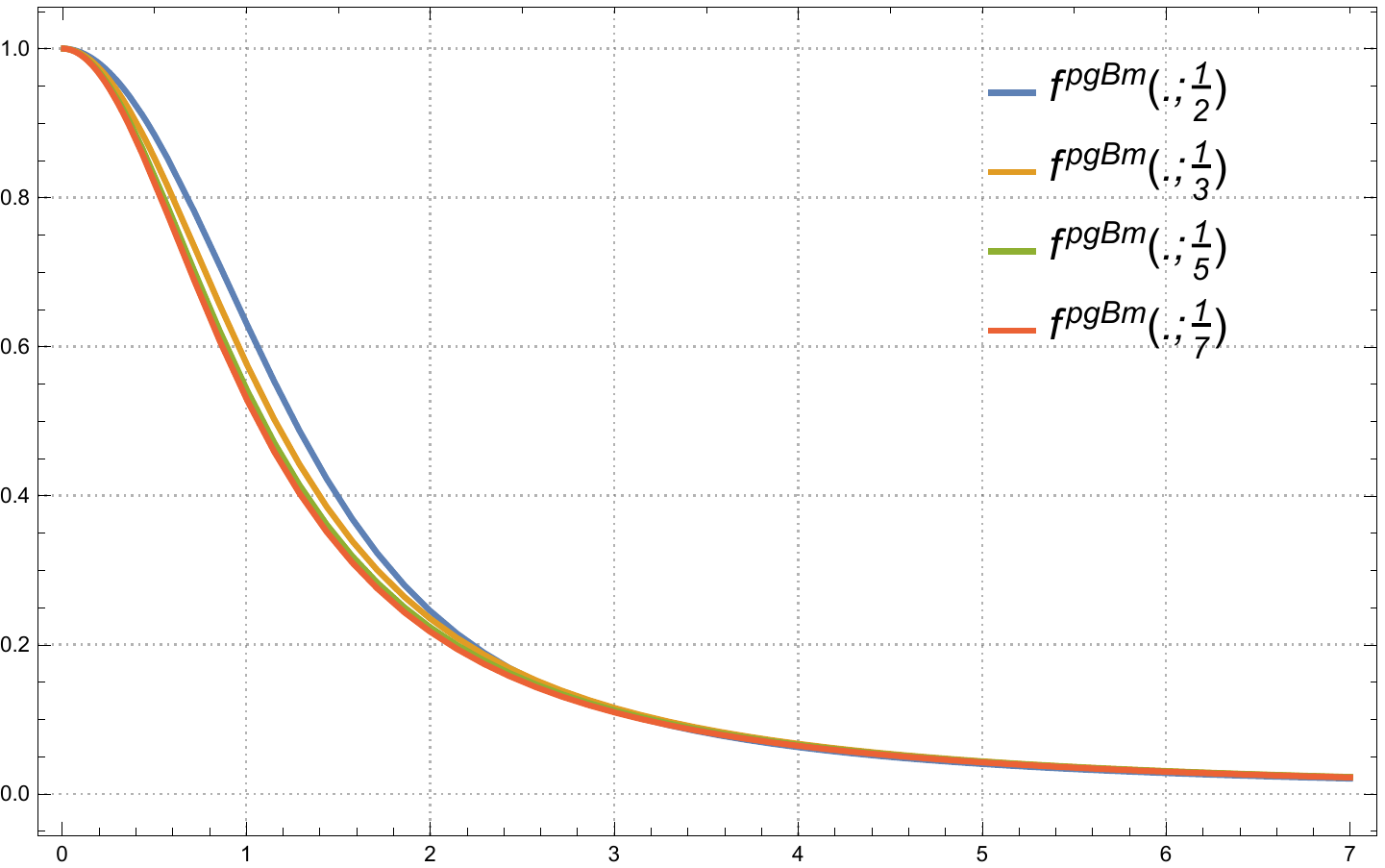}
\par\end{centering}
}\hfill{}\subfloat[\label{fig:LogLog-scale-1}LogLog scale.]{\begin{centering}
\includegraphics[scale=0.43]{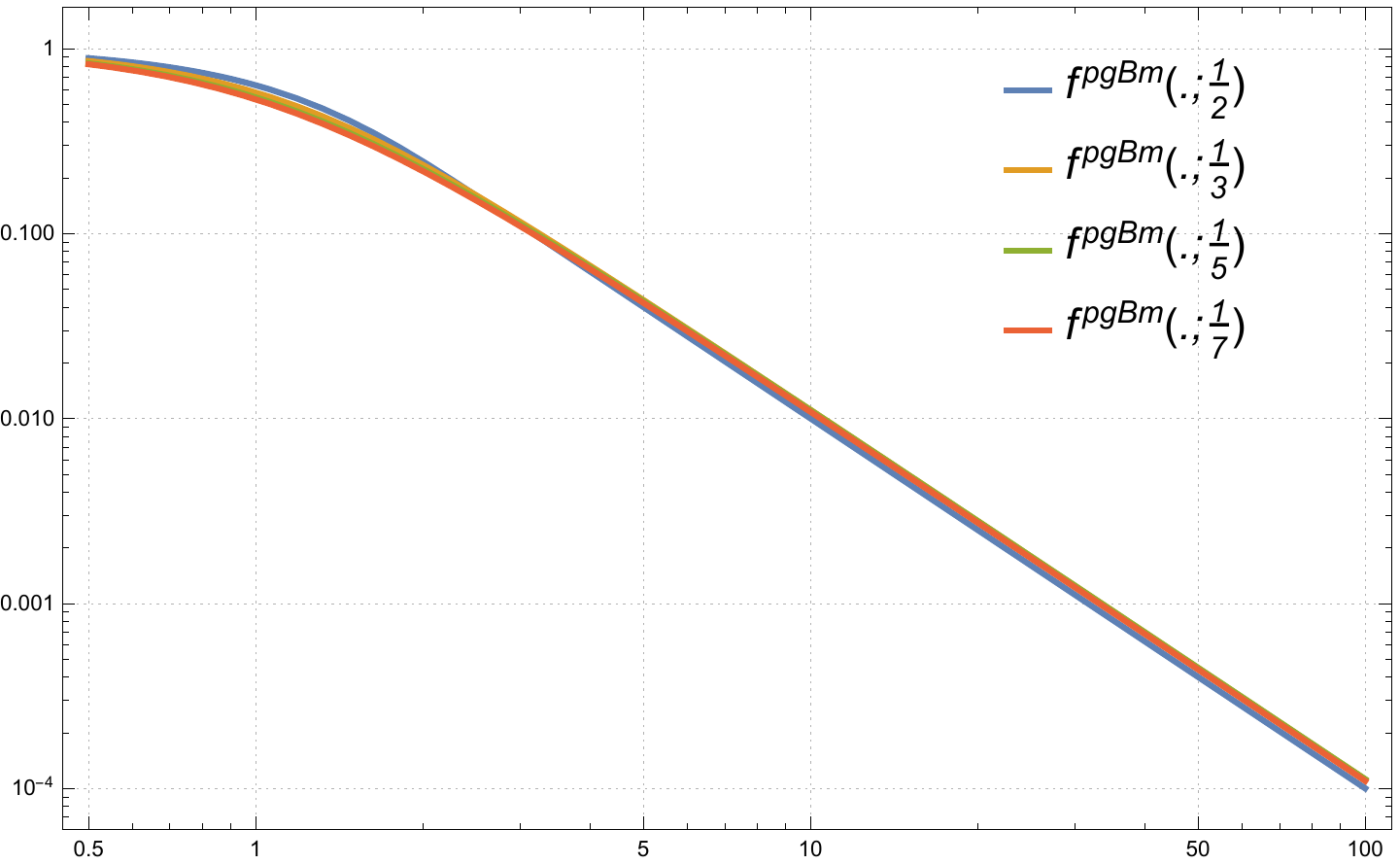}
\par\end{centering}
}
\par\end{centering}
\caption{\label{fig:forms_factors_Streit}Debye functions for the pgBm process
$X_{p}^{H}$ for $H=\frac{1}{2},\frac{1}{3},\frac{1}{5},\frac{1}{7}$.}
\end{figure}

For the the illustration of the dependence on large $y$, Figure \ref{fig:Kratky_gbm} shows the Kratky plot for pgBm. Here the asymptotical behavior is very well visible. 

\begin{figure}
\begin{centering}
\includegraphics[width=.5\textwidth]{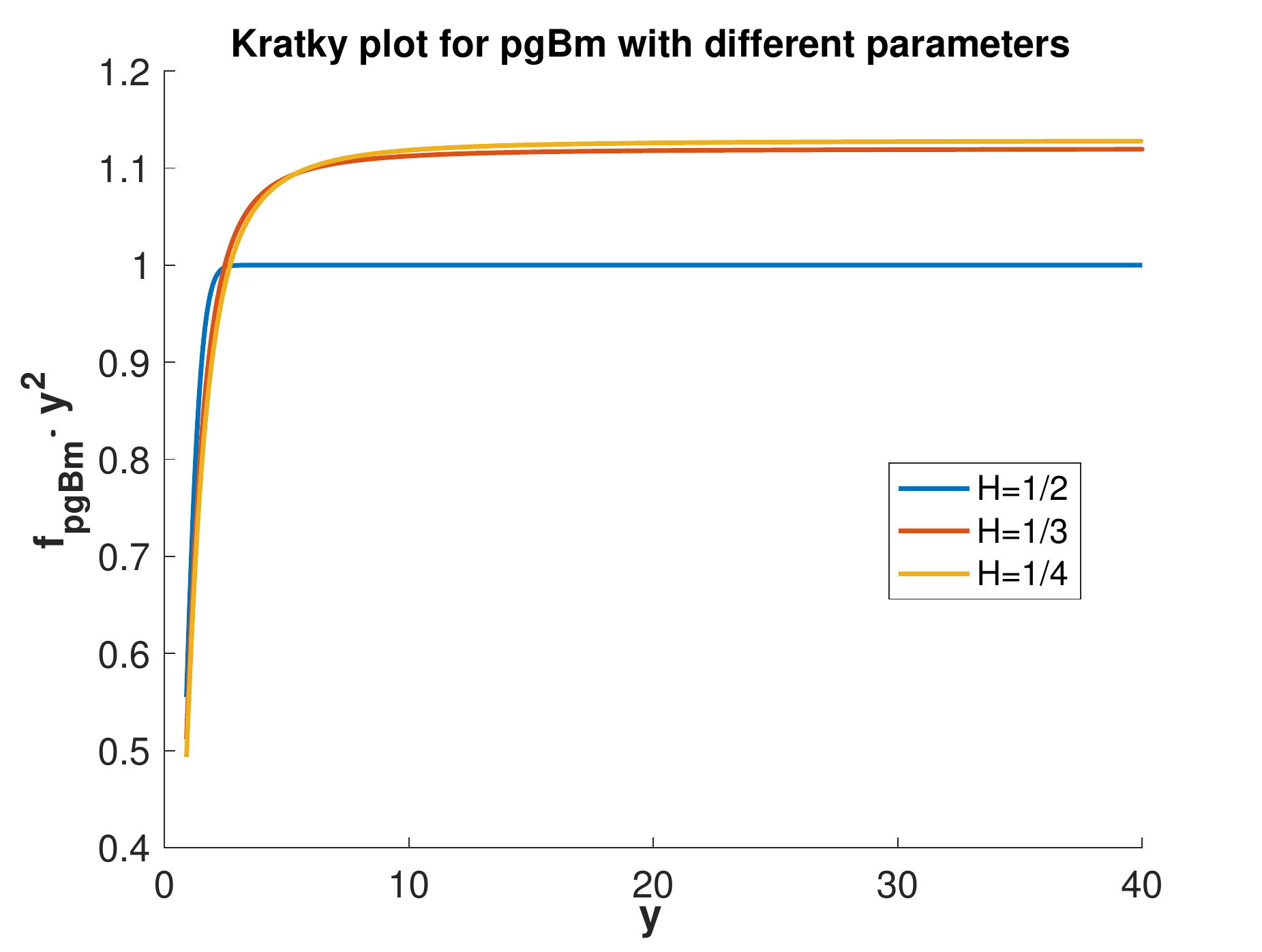}
\caption{\label{fig:Kratky_gbm}Kratky plot for the pgBm process
$X_{p}^{H}$ for $H=\frac{1}{2},\frac{1}{3},\frac{1}{4},\frac{1}{5}$.}
\end{centering}
\end{figure}
We have the following relation between the end--to-halftime length and
the radius of gyration for pgBm 
\[
\frac{\left(R_{e}^{\mathrm{pgBm}}\left(\frac{L}{2}\right)\right)^{2}}{2(2H+1)}=\left(R_{g}^{\mathrm{pgBm}}(L)\right)^{2}.
\]

\subsection{Form Factors for Periodic Generalized Grey Brownian Motion}

In this subsection we compute the form factor for the general class
of pggBm process $X_{p}^{\beta,H}$ introduced in Subsection\ \ref{subsec:pggBm}.
It corresponds to a generalization of the results obtained above for
the pfBm and pgBm processes.

The form factor for pggBm is given, for any $k\in\mathbb{R}$, by
\begin{eqnarray*}
S^{\mathrm{pggBm}}(k): & =\frac{2}{L^{2}}\int_{0}^{L}\int_{0}^{t}\mathbb{E}\left(e^{ik(X_{p}^{\beta,H}(t)-X_{p}^{\beta,H}(s))}\right)dsdt\\
 & =\frac{2}{L^{2}}\int_{0}^{L}\int_{0}^{t}\int_{0}^{\infty}M_{\beta}(\tau)\mathbb{E}\left(e^{-\frac{1}{2}k^{2}\sqrt{\tau}(B_{p}^{H}(t)-B_{p}^{H}(s))}\right)\,d\tau\,ds\,dt.\\
 & =\frac{2}{L^{2}}\int_{0}^{L}\int_{0}^{t}\int_{0}^{\infty}M_{\beta}(\tau)e^{-\frac{1}{2}k^{2}\tau d_{H}(t-s;L)}\,d\tau\,ds\,dt,
\end{eqnarray*}
Applying the Fubini theorem and making the change of variables $r=t-s$,
yields 
\begin{eqnarray*}
S^{\mathrm{pggBm}}(k) & =\frac{2}{L^{2}}\int_{0}^{\infty}M_{\beta}(\tau)\int_{0}^{L}\int_{0}^{t}e^{-\frac{1}{2}k^{2}\tau d_{H}(t-s;L)}\,d\tau\,ds\,dt\\
 & =\frac{2}{L^{2}}\int_{0}^{\infty}M_{\beta}(\tau)\int_{0}^{L}(L-r)e^{-\frac{1}{2}k^{2}\tau d_{H}(r;L)}\,dr\,d\tau.
\end{eqnarray*}
Once more Fubini's
theorem yields 
\begin{eqnarray*}
S^{\mathrm{pggBm}}(k) & =\frac{2}{L^{2}}\int_{0}^{L}(L-r)\int_{0}^{\infty}M_{\beta}(\tau)e^{-\frac{1}{2}k^{2}\tau d_{H}(r;L)}\,d\tau\,dr\\
 & =\frac{2}{L^{2}}\int_{0}^{L}(L-r)E_{\beta}\left(-\frac{k^{2}}{2}d_{H}(r;L)\right)\,dr\\
 & =\int_{0}^{1}E_{\beta}\left(-\frac{k^{2}}{2}\left(\frac{L}{2}\right)^{2H}r^{2H}\right)\,dr.
\end{eqnarray*}
Using the equality (\ref{eq:Euler_transform-MLf}) with $\rho=\sigma=\alpha=1$
and $\gamma=2H$ we obtain 
\begin{eqnarray*}
S^{\mathrm{pggBm}}(k) & =\,_{2}\Psi_{2}\left(\begin{array}{cc}
(1,2H), & (1,1)\\
(1,\beta), & (2,2H)
\end{array}\bigg|-\frac{k^{2}}{2}\left(\frac{L}{2}\right)^{2H}\right)\\
 & =\sum_{n=0}^{\infty}\frac{1}{(1+2Hn)\Gamma(1+\beta n)}\left(-\frac{k^{2}}{2}\left(\frac{L}{2}\right)^{2H}\right)^{n}.
\end{eqnarray*}
The Debye function for $X_{p}^{\beta,H}$ is obtained, denoting $y^{2}=\frac{k^{2}}{2}\left(\frac{L}{2}\right)^{2H}$,
as 
\[
f^{\mathrm{pggBm}}(y;\beta,H)=\sum_{n=0}^{\infty}\frac{\big(-y^{2}\big)^{n}}{(1+2Hn)\Gamma(1+\beta n)}.
\]
In Figure\ \ref{fig:Debye-pggBm} we plot the (truncated at $n=700$)
Debye function of $X_{p}^{\beta,H}$ for $\beta=\frac{1}{2}$ and
$H=\frac{1}{2},\frac{1}{3},\frac{1}{5},\frac{1}{7}$ in linear scale
Figure\ \ref{fig:Debye-pggBm}-\subref{fig:pggBm-Linear-scale} and
LogLog scale Figure\ \ref{fig:Debye-pggBm}-\subref{fig:pggBm-LogLog-scale}.

The radius of gyration for pggBm is obtained by expanding the form
factor to lower order 
\[
\big(R_{g}^{\beta,H}(L)\big)^{2}=\frac{L^{2H}}{2^{2H+1}(2H+1)\Gamma(\beta+1)}
\]
and the end-to-halftime length with time parameter $t\in[0,\frac{L}{2}]$
may be computed using (\ref{eq:2nd-moment-inc-pggBm}). We obtain
\[
\left(R_{e}^{\beta,H}\left(\frac{L}{2}\right)\right)^{2}=\mathbb{E}\left(\left(X_{p}^{\beta,H}\left(\frac{L}{2}\right)\right)^{2}\right)=\frac{d_{H}\left(\frac{L}{2};L\right)}{\Gamma(\beta+1)}=\frac{L^{2H}}{\Gamma(\beta+1)2^{2H}}.
\]
As a result, the following relation holds 
\[
\frac{\left(R_{e}^{pggBm}\left(\frac{L}{2}\right)\right)^{2}}{2(2H+1)}=\big(R_{g}^{\beta,H}(L)\big)^{2}.
\]

\begin{figure}
\begin{centering}
\subfloat[\label{fig:pggBm-Linear-scale}Linear scale.]{\begin{centering}
\includegraphics[scale=0.43]{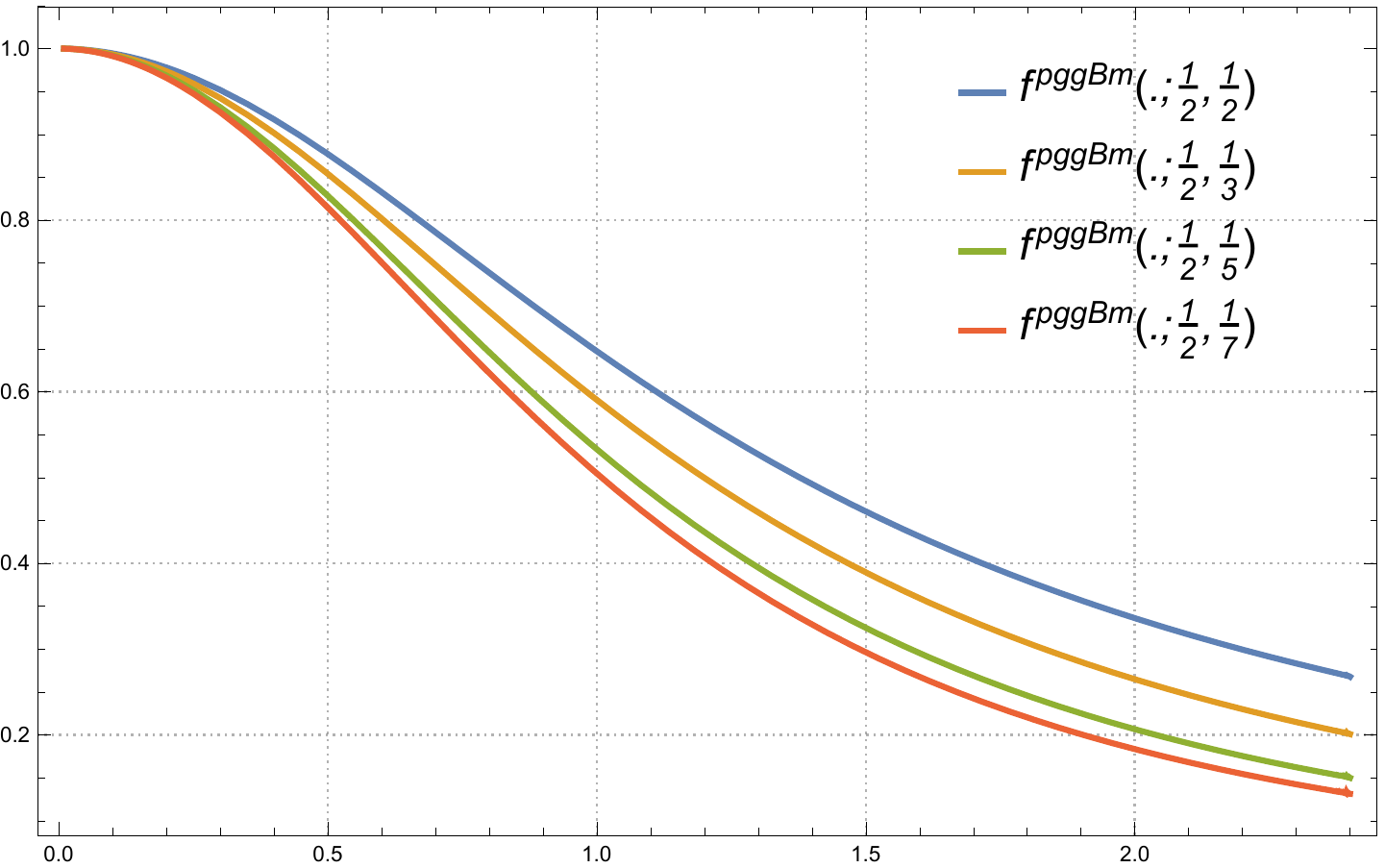}
\par\end{centering}
}\hfill{}\subfloat[\label{fig:pggBm-LogLog-scale}LogLog scale.]{\begin{centering}
\includegraphics[scale=0.31]{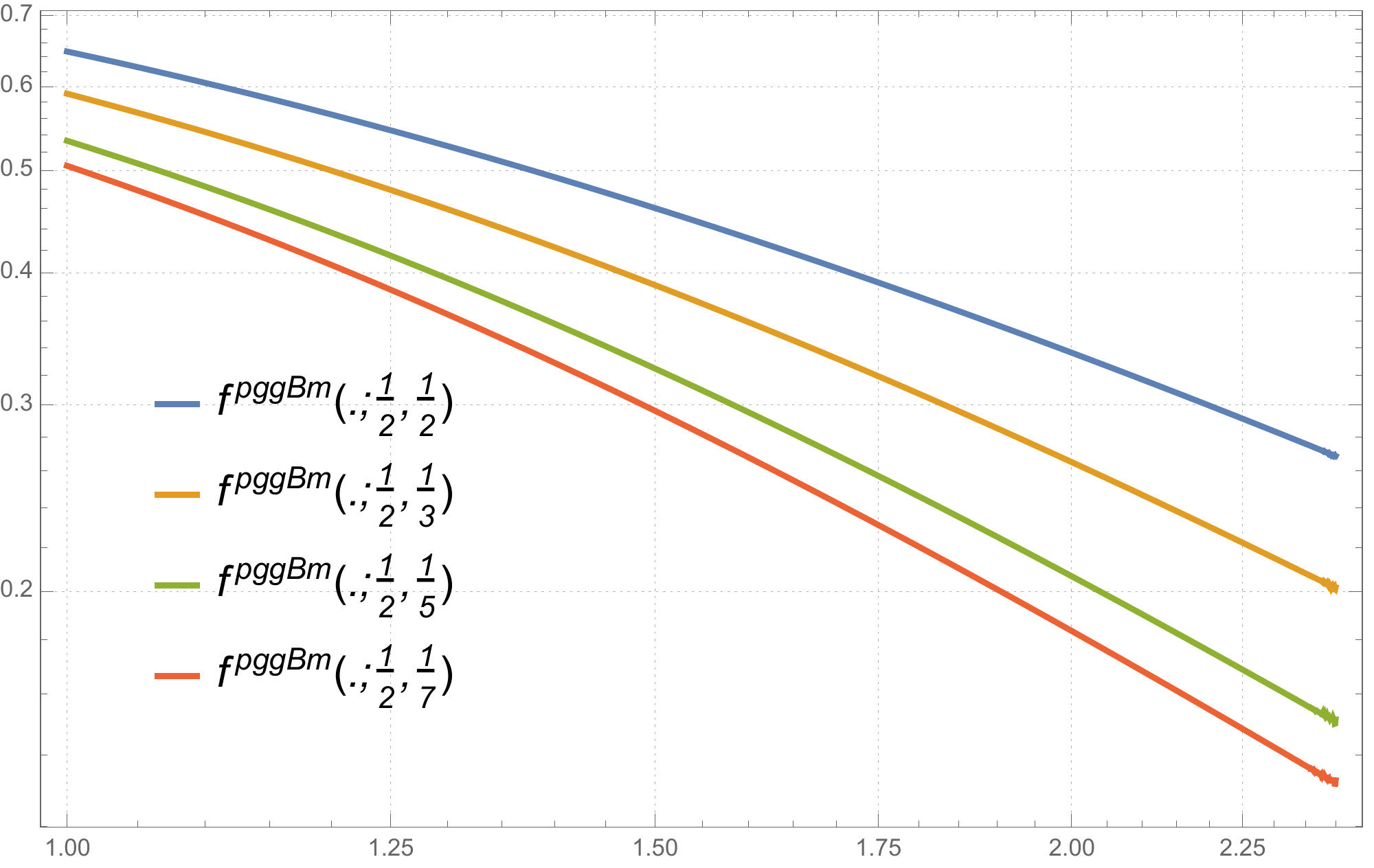}
\par\end{centering}
}
\par\end{centering}
\caption{\label{fig:Debye-pggBm}Debye functions for the pggBm process $X_{p}^{\beta,H}$
for $\beta=\frac{1}{2}$ and $H=\frac{1}{2},\frac{1}{3},\frac{1}{5},\frac{1}{7}$.}
\end{figure}

\appendix

\section{The Mittag-Leffler and $M$-Wright functions}

\label{sec:appendix}In this appendix we introduce two families of
functions which are used in this paper. They are the family of generalized
Mittag-Leffler functions $E_{\beta,\alpha}$, $0<\beta\le1$, $\alpha\in\mathbb{C}$
and the family of $M$-Wright functions $M_{\beta}$, $0<\beta\le1$
which is a special case of the Wright functions $W_{\lambda,\rho}$,
$\lambda>-1,\;\mu\in\mathbb{C}$. More details and these classes of
functions may found in \cite{GKMS2014} and references therein.

The Mittag-Leffler function was introduced by G.\ Mittag-Leffler
in a series of papers \cite{Mittag-Leffler1903,Mittag-Leffler1904,Mittag-Leffler1905}.
\begin{defn}[Mittag-Leffler function]
\label{def:MLf}
\begin{enumerate}
\item For $\beta>0$ the Mittag-Leffler function $E_{\beta}$ is defined
as an entire function by the following series representation 
\begin{equation}
E_{\beta}(z):=\sum_{n=0}^{\infty}\frac{z^{n}}{\Gamma(\beta n+1)},\quad z\in\mathbb{C},\label{eq:MLf}
\end{equation}
where $\Gamma$ denotes the gamma function.
\item For any $\rho\in\mathbb{C}$ the generalized Mittag-Leffler function
is an entire function defined by its power series 
\begin{equation}
E_{\beta,\rho}(z):=\sum_{n=0}^{\infty}\frac{z^{n}}{\Gamma(\beta n+\rho)},\quad z\in\mathbb{C}.\label{eq:gMLf}
\end{equation}
Note the relation $E_{\beta,1}(z)=E_{\beta}(z)$ and $E_{1}(z)=e^{z}$
for any $z\in\mathbb{C}$.
\end{enumerate}
\end{defn}

We have the following asymptotic for the generalized Mittag-Leffler
function $E_{\beta,\alpha}$.
\begin{prop}[{{{cf.\ \protect\cite[Section~4.7]{GKMS2014}}}}]
\label{prop:asymp-gMLf}Let $0<\beta<2$, $\alpha\in\emph{C}$ and
$\delta$ be such that 
\[
\frac{\beta\pi}{2}<\delta<\min\{\pi,\beta\pi\}.
\]
Then, for any $m\in\emph{N}$, the following asymptotic formulas hold:
\begin{enumerate}
\item If $|\arg(z)|\leq\delta$, then 
\[
E_{\beta,\alpha}(z)=\frac{1}{\beta}z^{(1-\alpha)/\beta}\exp(z^{1/\beta})-\sum_{n=1}^{m}\frac{z^{-n}}{\Gamma(\alpha-\beta n)}+O(|z|^{-m-1}),\;|z|\rightarrow\infty.
\]
\item \label{asym-gMLf-two}If $\delta\le|\arg(z)|\leq\pi$, then 
\begin{equation}
E_{\beta,\alpha}(z)=-\sum_{n=1}^{m}\frac{z^{-n}}{\Gamma(\alpha-\beta n)}+O(|z|^{-m-1}),\quad|z|\rightarrow\infty.\label{eq:asymptotic_MLf}
\end{equation}
\end{enumerate}
\end{prop}

The Euler integral transform of the MLf may be used to compute the
following integral, with real parts $\mathrm{Re}(\alpha),\mathrm{Re}(\beta),\mathrm{Re}(\sigma)>0$,
$\rho\in\mathbb{C}$ and $\gamma>0$, cf.\ \cite[eq.~(2.2.13)]{Mathai2008}
\begin{equation}
\int_{0}^{1}t^{\rho-1}(1-t)^{\sigma-1}E_{\beta,\alpha}(xt^{\gamma})\,dt=\Gamma(\sigma){}_{2}\Psi_{2}\left(\begin{array}{cc}
(\rho,\gamma), & (1,1)\\
(\alpha,\beta), & (\sigma+\rho,\gamma)
\end{array}\bigg|x\right),\label{eq:Euler_transform-MLf}
\end{equation}
where $_{2}\Psi_{2}$ is the Fox-Wright function (also called generalized
Wright function \cite{Kilbas2002}, \cite[Appendix~F, eq.~(F.2.14)]{GKMS2014}
and \cite{MP07}) given for $x,a_{i},c_{i}\in\mathbb{C}$ and $b_{i},d_{i}\in\mathbb{R}$
by 
\[
_{2}\Psi_{2}\left(\begin{array}{cc}
(a_{1},b_{1}), & (a_{2},b_{2})\\
(c_{1},d_{1}), & (c_{2},d_{2})
\end{array}\bigg|x\right)=\sum_{n=0}^{\infty}\frac{\Gamma(a_{1}+b_{1}n)\Gamma(a_{2}+b_{2}n)}{\Gamma(c_{1}+d_{1}n)\Gamma(c_{2}+d_{2}n)}\frac{x^{n}}{n!}.
\]
In particular, when $\rho=\alpha$ and $\gamma=\beta$, eq.\ (\ref{eq:Euler_transform-MLf})
simplifies to 
\begin{equation}
\int_{0}^{1}t^{\alpha-1}(1-t)^{\sigma-1}E_{\beta,\alpha}(xt^{\beta})\,dt=\Gamma(\sigma)E_{\beta,\alpha+\sigma}(x).\label{eq:Euler_transform-MLf1}
\end{equation}
Both integrals (\ref{eq:Euler_transform-MLf}) and (\ref{eq:Euler_transform-MLf1})
appears in computing the form factors in Section~\ref{sec:Form_Factors}.

The Wright function is defined by the following series representation
which converges in the whole complex $z$-plane 
\[
W_{\lambda,\mu}(z):=\sum_{n=0}^{\infty}\frac{z^{n}}{n!\Gamma(\lambda n+\mu)},\quad\lambda>-1,\;\mu\in\mathbb{C}.
\]
An important particular case of the Wright function is the so called
$M$-Wright function $M_{\beta}$, $0<\beta\le1$ (in one variable)
defined by 
\begin{equation}
M_{\beta}(z):=W_{-\beta,1-\beta}(-z)=\sum_{n=0}^{\infty}\frac{(-z)^{n}}{n!\Gamma(-\beta n+1-\beta)}.\label{eq:M-Wright-function}
\end{equation}
For the choice $\beta=\frac{1}{2}$ the corresponding $M$-Wright
function reduces to the Gaussian density 
\begin{equation}
M_{\frac{1}{2}}(z)=\frac{1}{\sqrt{\pi}}\exp\left(-\frac{z^{2}}{4}\right).\label{eq:MWright_Gaussian}
\end{equation}

The MLf $E_{\beta}$ and the $M$-Wright are related through the Laplace
transform 
\begin{equation}
\int_{0}^{\infty}e^{-s\tau}M_{\beta}(\tau)\,d\tau=E_{\beta}(-s).\label{eq:LaplaceT_MWf}
\end{equation}

\subsection*{Acknowledgement}

Financial support from FCT -- Funda{\c c\~ao} para a Ci{\^e}ncia
e a Tecnologia through the project UID/MAT/04674/2019 (CIMA Universidade
da Madeira) is gratefully acknowledged.

\vspace{.5cm}
\bibliographystyle{ieeetr}

\end{document}